\documentclass[a4paper,11pt]{amsart}

\pdfoutput=1 

\usepackage{definitions}

\hypersetup{
	pdftitle={integrable systems},
	pdfsubject={High Energy Physics},
	pdfauthor={Thiago Araujo},
	pdfkeywords={gauge; susy; strings; fields; cft},
	pdfsubject={hep-th},
	colorlinks=true,linkcolor=link,citecolor=link,urlcolor=link,linktocpage
}

\begin{document}

%%%%%%%%%%%%%%%%%%%%%%%%%%%%%%%%%%%%%%%%%%%%%%
%%%%%%%%%%%%%%%%%%%%%%%%%%%%%%%%%%%%%%%%%%%%%%
%%%%%%%%%%%%%%%%%%%%%%%%%%%%%%%%%%%%%%%%%%%%%%
%%%%%%%%%%%%%%%%%%%%%%%%%%%%%%%%%%%%%%%%%%%%%%

\title[Comments on Slavnov products, Temperley-Lieb open spin chains, and KP tau functions]{Comments on Slavnov products, Temperley-Lieb open spin chains, and KP tau functions}

\author{Thiago Araujo}

\address{\noindent Albert Einstein Center for Fundamental Physics, Institute for Theoretical Physics, University of Bern, Sidlerstrasse 5, CH-3012, Bern, Switzerland}
\email{\texttt{\href{thgr.araujo@gmail.com}{thgr.araujo@gmail.com}}} 

%\keywords{Temperley-Lieb, XXZ spin chain, KP, matrix integrals}
%\date{\today}

\begin{abstract}
We study Slavnov's inner products for open Temperley-Lieb chains and their relations with the KP hierarchy. We show that when \(s=1/2\) the quantum group invariant XXZ spin chain has Slavnov products given by the quotient of tau functions. The Schur polynomials expansion and a matrix construction for the Baker-Akhiezer functions are also briefly considered.

\bigskip

\end{abstract}

\dedicatory{To the memory of my father, David Araujo}

\maketitle

\setcounter{tocdepth}{1}
\tableofcontents

%%%%%%%%%%%%%%%%%%%%%%%%%%%%%%%%%%%%%%%%%%%%%%
%%%%%%%%%%%%%%%%%%%%%%%%%%%%%%%%%%%%%%%%%%%%%%
%%%%%%%%%%%%%%%%%%%%%%%%%%%%%%%%%%%%%%%%%%%%%%
%%%%%%%%%%%%%%%%%%%%%%%%%%%%%%%%%%%%%%%%%%%%%%
\section{Introduction}

Classical and quantum integrable systems share an extensive set of mathematical commonalities. The quantization of a classical integrable system is, at least in principle, the proper quantization of their defining integrable structures. Remarkably, a few different relations between classical and quantum integrable systems have been observed~\cite{Wu:1975mw, Its:1992bj, Korepin:1993kvr, Korepin2000, Alexandrov:2011aa, Foda2009a, Foda2009}, but the underlying mechanisms (if any) governing these correspondences are not known yet. 

The present work deals with the unexpected relations originally studied by Foda, Wheeler, and Zuparic in~\cite{Foda2009a, Foda2009}. The authors prove that for the XXZ spin-\(\sfrac{1}{2}\) chain with periodic boundary conditions, the Slavnov's inner product~\cite{Slavnov1989} is (modulo multiplicative terms) a tau-function of the \emph{Kadomtsev-Petviashvili} (KP) integrable hierarchy~\cite{Miwa2000}. More specifically, using the Jacobi-Trudi identities and the Schur polynomials expansions -- important components in their work -- the authors have shown that the inner product between on-shell and off-shell Bethe states can be represented in terms of a solution of an integrable hierarchy.

This duality might allow us to study quantum spin chain from another perspective, but one might ask first whether the results of~\cite{Foda2009a, Foda2009} are beautiful occurrences of the XXZ spin-\(\sfrac{1}{2}\) chain with periodic boundary conditions, or if there is something more interesting behind the curtains. In particular, we would like to know if this duality holds for other boundary conditions, representations \(s>\sfrac{1}{2}\), or other types of spin chains. These bold questions find two difficult obstacles. 

First, Slavnov products are complex objects to construct, see~\cite{Korepin:1993kvr}, and there is no general protocol teaching us how to build them; it is safe to say that we are essentially confined to the case-by-case analysis at the time of writing. Despite these difficulties, some new results have recently appeared in the literature~\cite{Nepomechie:2016ejv, Nepomechie:2016ruv} allowing us to advance this program a bit further. These new results are the examples we address in the current work. 

Furthermore, the tools used in~\cite{Foda2009a, Foda2009} are very specific to their particular case, and hardly generalizable to other contexts. More specifically, the exceptional structure of the Slavnov product allowed the authors to use some  Jacobi-Trudi-like identities in their study, but this feature is not expected in other cases, see the structure of~\cite{Nepomechie:2016ejv, Nepomechie:2016ruv}. In the current work, we also try to find some new ingredients and tools to study this duality between spin chains and classical integrable hierarchies.

In this work we argue that the Slavnov product for the spin-\(s\) \ac{tl} open chain~\cite{Nepomechie:2016ejv, Nepomechie:2016ruv} defines two KP tau functions. Section~\ref{sec:tlchain} starts with a short review on the \ac{tl} spin-\(s\) chain, and in sections~\ref{sec:slavnov} and~\ref{sec:schur}, we show how a pair of KP tau functions emerges from the Slavnov product, and we study some of the most immediate consequences. We conclude in section~\ref{sec:remarks} with further remarks, a comparison between our results and~\cite{Foda2009a, Foda2009}, and some open problems.

%%%%%%%%%%%%%%%%%%%%%%%%%%%%%%%%%%%%%%%%%%%%%%
%%%%%%%%%%%%%%%%%%%%%%%%%%%%%%%%%%%%%%%%%%%%%%
%%%%%%%%%%%%%%%%%%%%%%%%%%%%%%%%%%%%%%%%%%%%%%
%%%%%%%%%%%%%%%%%%%%%%%%%%%%%%%%%%%%%%%%%%%%%%

\section{Temperley-Lieb Spin Chains}
\label{sec:tlchain}

This section introduces the \ac{tl} open spin-\(s\) chain following Nepomechie and Pimenta~\cite{Nepomechie:2016ejv, Nepomechie:2016ruv}, see also~\cite{Temperley1971}. The \ac{tl} algebra \(\mathcal{A}_N(\varepsilon)\) is defined by the generators \(\{ U_{i} \ |\ i=1, \dots, N-1 \}\) and the relations
\be 
U_{i}^2= \varepsilon U_{i} \; ,\quad U_{i} U_{i\pm 1}U_{i} = U_{i} \; ,\quad U_{i} U_{j}=U_{i} U_{j} \ \ \   |i-j|>1\; ,
\ee
where we write \(\varepsilon = -\left(q+q^{-1}\right) \). The generators can be diagrammatically represented as in figure~\ref{fig:generators}, in such a way that the algebraic relations can be elegantly represented as braiding operations, see~\cite{Abramsky2009, Doikou:2009xq} for reviews. 
\begin{figure}[h!]
	\centering
	\begin{tikzpicture}[thick, scale=0.8]
		\node at (-2.5,0) {\(\mathbb{1}=\)};
		\draw[line width=0.5mm] (-2,-1) -- (-2,1); \node at (-2,-1) {\(\bullet\)}; \node at (-2,1) {\(\bullet\)}; \node at (-2,1.3) {\(1\)};
		\draw[line width=0.5mm] (-1,-1) -- (-1,1); \node at (-1,-1) {\(\bullet\)}; \node at (-1,1) {\(\bullet\)};  \node at (-1,1.3) {\(2\)}; 
		\node at (-0.25,0) {\(\cdots\)};
		\draw[line width=0.5mm] (0.5,-1) -- (0.5,1); \node at (0.5,-1) {\(\bullet\)}; \node at (0.5,1) {\(\bullet\)};  \node at (0.5,1.3) {\(N-1\)}; 
		\draw[line width=0.5mm] (1.5,-1) -- (1.5,1); \node at (1.5,-1) {\(\bullet\)}; \node at (1.5,1) {\(\bullet\)};  \node at (1.5,1.3) {\(N\)}; 
	\end{tikzpicture}\hspace{1.0cm}
	\begin{tikzpicture}[thick, scale=0.8]
		\node at (-2.8,0) {\(U_i=\)};
		\draw[line width=0.5mm] (-2,-1) -- (-2,1); \node at (-2,-1) {\(\bullet\)}; \node at (-2,1) {\(\bullet\)}; \node at (-2,1.3) {\(1\)};
		\draw[line width=0.5mm] (-1,-1) -- (-1,1); \node at (-1,-1) {\(\bullet\)}; \node at (-1,1) {\(\bullet\)};  \node at (-1,1.3) {\(2\)}; 
		\node at (-0.25,0) {\(\cdots\)};
		\draw[line width=0.5mm] (1.5,1) arc (0:-180:0.5); \node at (0.5,-1) {\(\bullet\)}; \node at (0.5,1) {\(\bullet\)};  \node at (0.5,1.3) {\(i\)}; 
		\draw[line width=0.5mm] (1.5,-1) arc (0:180:0.5); \node at (1.5,-1) {\(\bullet\)}; \node at (1.5,1) {\(\bullet\)};  \node at (1.5,1.3) {\(i+1\)}; 
		\node at (2.25,0) {\(\cdots\)};
		\draw[line width=0.5mm] (3,-1) -- (3,1); \node at (3,-1) {\(\bullet\)}; \node at (3,1) {\(\bullet\)};  \node at (3,1.3) {\(N-1\)}; 
		\draw[line width=0.5mm] (4,-1) -- (4,1); \node at (4,-1) {\(\bullet\)}; \node at (4,1) {\(\bullet\)};  \node at (4,1.3) {\(N\)}; 
	\end{tikzpicture}
	\caption{Identity and operators \(U_i\), \(i=1,\dots, N-1\), in \(\mathcal{A}_N\).}
	\label{fig:generators}
\end{figure}
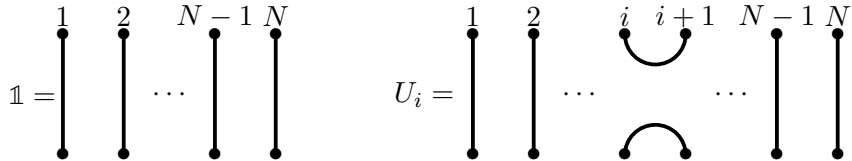

Despite its simplicity, this structure appears in a wide range of settings~\cite{Abramsky2009}; and we are particularly interested in spin chains defined by these generators.

\subsection{Temperley-Lieb spin chains}

The length \(N\) open \ac{tl} spin chain~\cite{Nepomechie:2016ejv, Nepomechie:2016ruv} with free boundary conditions is defined by the  Hamiltonian
\be 
\label{eq:tlhamil}
\mathcal{H}=\sum_{i=1}^{N-1} U_{i}\; .
\ee
This spin chain is quantum integrable for any spin-\(s\) representation of \(U_q(\mathfrak{sl}(2))\), and the case \(s=1/2\) reduces to the \(U_q(\mathfrak{sl}(2))\)-invariant spin chain~\cite{Pasquier:1989kd}. The R-matrix is
\be 
R(v) = w(vq) \mathcal{P} + w(v) \mathcal{P}U\; ,
\ee
where \(\mathcal{P}\) is the permutation operator in \(\mathbb{C}^{2s + 1}\otimes \mathbb{C}^{2s + 1}\), \(U\) is related to the algebra generators via \(U_i = \mathbb{1}^{\otimes (i-1)}\otimes U\otimes   \mathbb{1}^{\otimes (N-i-1)}\), and we have also used the notation
\be 
w(v) = v - v^{-1}\; .
\ee

The transfer matrix is written as
\be 
t(v):=\Tr_0 \left( K_0^+(v) T_0(v) K_0^-(v) \hat{T}_0(v)  \right)
\ee
where \(T_0(v)\) and \( \hat{T}_0(v)\) are the monodromy matrices, and \(K^\pm(v)\) satisfy the boundary Yang-Baxter equations. Explicitly, the matrices \(K^\pm(v)\) are
\be 
K^- = \mathbb{1}\qquad 
K^+ = \textrm{diag}\left( Q^{-2s}, Q^{-2(s-1)}, \dots, Q^{2s} \right)\; , 
\ee
with \(\sum_{k=-s}^s Q^{2k} \equiv -( q+q^{-1})\)\; .

The eigenvalues \(\Lambda(v)\equiv \Lambda(v, \bm{u})\) of \(t(v)\) are 
\be 
\label{eq:eigenvalues}
\Lambda(v) = -\frac{1}{w(v^2 q)} \left( w(v^2 q^2) w(v q)^{2N} \prod_{j=1}^M \frac{w(vq^{-1}/u_j) w(v u_j) }{ w(v/u_j) w(v q u_j)} + w(v^2) w(v)^{2N}   \prod_{j=1}^M \frac{w(vq/u_j) w(vq^2 u_j) }{ w(v/u_j) w(v q u_j)} \right)\; ,
\ee
where \(\bm{u}\equiv \{u_k \ | \ k=1, \dots, M \}\) are the Bethe roots in the \(M\)-magnon sector. 
\medskip

Consider the Temperley-Lieb spin chain~(\ref{eq:tlhamil}) with \(s=\sfrac{1}{2}\), and the states
\be 
\ket{\Psi_{bethe}(\bm{u})} =\prod_{k=1}^{M}\mathcal{B}(u_k)|0\rangle \qquad
\ket{\Psi(\bm{v})} =\prod_{k=1}^{M}\mathcal{B}(v_k)|0\rangle 
\ee
where \(\bm{u}\) satisfy the Bethe equations, and \(\bm{v}\) are free parameters. It has been proved~\cite{Kitanine:2007bi, Wang:2002paz}, and \cite{Nepomechie:2016ejv, Nepomechie:2016ruv}, that the Slavnov product~\cite{Slavnov1989} is given by
\be 
\label{eq:slavnov} 
\bracket{\Psi_{bethe}(\bm{u})}{\Psi(\bm{v})} =
2^{-M}Q^{-2Ms} \prod_{j=1}^M \frac{w(u_j)^{2N} u_j w(u_j^2)}{w(u_j^2) w(v_j^2 q^2)} \prod_{\substack{i, j=1 \\ i>j}}^M \frac{w(u_iu_j q^2)}{w(u_i u_j)}
\frac{\det\limits_{i,j} \partial_{u_i} \Lambda(v_j, \bm{u})}{\det\limits_{i,j} \left( \frac{1}{w(v_i u_j^{-1}) w(v_i u_j q)} \right)}
\; .
\ee
In~\cite{Nepomechie:2016ejv, Nepomechie:2016ruv}, a conjecture concerning this inner product establishes that formula~(\ref{eq:slavnov}) is also valid for any spin-\(s\) representation of \(U_q(\mathfrak{sl}(2))\). Although we will assume that this conjecture is true, it is important to make the distinction between the \(s=\sfrac{1}{2}\) and \(s>\sfrac{1}{2}\) cases.

The results just described are the ingredients we need to advance with this work, and we refer to~\cite{Nepomechie:2016ejv, Nepomechie:2016ruv} for further details.

%%%%%%%%%%%%%%%%%%%%%%%%%%%%%%%%%%%%%%%%%%%%%%
%%%%%%%%%%%%%%%%%%%%%%%%%%%%%%%%%%%%%%%%%%%%%%
%%%%%%%%%%%%%%%%%%%%%%%%%%%%%%%%%%%%%%%%%%%%%%
%%%%%%%%%%%%%%%%%%%%%%%%%%%%%%%%%%%%%%%%%%%%%%

\section{Slavnov product and KP tau functions}
\label{sec:slavnov}

In this section we show that the Slavnov product~(\ref{eq:slavnov}) can be written in terms of a quotient of KP tau functions. We also study some immediate consequences of this result, including a matrix integral representation for a normalized inner product.

\subsection{KP tau functions and matrix integrals}

Let us start with the main result of this work. 

\begin{theorem}
The Slavnov product~(\ref{eq:slavnov}) for the Temperley-Lieb spin-\(\sfrac{1}{2}\) chain is written as
\be 
\label{eq:sptheorem}
\bracket{\Psi_{bethe}(\bm{u})}{\Psi(\bm{v})} = G(\bm{v}, \bm{u})  \frac{\tau_{\bm{u}}^{(1)} (\bm{v})}{\tau_{\bm{u}}^{(2)} (\bm{v})}
\ee
where \(\tau_{\bm{u}}^{(1)}\) and \(\tau_{\bm{u}}^{(1)}\) are KP hierarchy tau functions. Assuming that the conjecture~\cite{Nepomechie:2016ejv, Nepomechie:2016ruv} is true, the theorem is immediately extended to all \(s>\sfrac{1}{2}\).
\end{theorem}	

\begin{proof}
As we mentioned before, Nepomechie and Pimenta~\cite{Nepomechie:2016ejv, Nepomechie:2016ruv} have conjectured that the Slavnov product is given by expression~(\ref{eq:slavnov}). Define the functions
\bse 
\be 
G(\bm{u}, \bm{v}) = 2^{-M}Q^{-2Ms} \prod_{j=1}^M \frac{w(u_j)^{2N} u_j w(u_j^2)}{w(u_j^2)} \prod_{\substack{i, j=1 \\ i>j}}^M \frac{w(u_iu_j q^2)}{w(u_i u_j) w(v_j^2 q^2)}\; ,
\ee
and 
\be 
\mathcal{K}^{-1}_{\bm{u}}(\bm{v}) = \frac{\det\limits_{i,j} \partial_{u_i} \Lambda(v_j, \bm{u})}{\det\limits_{i,j} \left( \frac{1}{w(v_i u_j^{-1}) w(v_i u_j q)} \right)}\; .
\ee
\ese
Observe that all interesting dependence on the free parameters \(\bm{v}\) is contained in \(\mathcal{K}_{\bm{u}}(\bm{v})\). With these definitions, one has
\be 
\bracket{\Psi_{bethe}(\bm{u})}{\Psi(\bm{v})} = G(\bm{v}, \bm{u}) \mathcal{K}^{-1}_{\bm{u}}(\bm{v})\; .
\ee
We may say that \(\mathcal{K}^{-1}_{\bm{u}}(\bm{v})\) is the \emph{kernel} of the Slavnov product.

Let us now write the kernel as
\be 
\mathcal{K}^{-1}_{\bm{u}}(\bm{v}) = \frac{\det\limits_{i,j} F^{(1)}_i(v_j)}{ \det\limits_{i,j} F^{(2)}_i(v_j) }\; ,
\ee	
where the set of functions \(\bm{F}^{(a)}(z)\), for \(a=1,2\), are defined by their components
\be 
\begin{split}
\bm{F}^{(1)}(z) & =\{F^{(1)}_1(z), \dots, F^{(1)}_M(z)\}\; : \qquad F^{(1)}_i(z) := \partial_{u_i} \Lambda(z, \bm{u})\\
\bm{F}^{(2)}(z) & =\{F^{(2)}_1(z), \dots, F^{(2)}_M(z)\}\; : \qquad F^{(2)}_i(z) := \frac{1}{w(z u_i^{-1}) w(z u_i q)} 
\end{split}	\; .
\ee

Using the Vandermonde determinant 
\[
\Delta(\bm{v})\equiv \det\limits_{i,j}(v_i^{M-j}) = \prod_{i<j} (x_i-x_j)\; ,
\] 
it is well known that the functions 
\be 
\label{eq:taufunc}
\tau^{(a)}(\bm{v}):= \frac{\det\limits_{i,j} F^{(a)}_i(v_j) }{\Delta(\bm{v})}\; ,\quad a=1,2\; ,
\ee
satisfy the~\ac{hbe} of the KP hierarchy~\cite{Kharchev:1991cy, Kharchev:1991cu, Morozov:1994hh, Marshakov:2006rh}
\be 
\label{eq:hbe}
\oint \frac{\dd z}{2\pi i} e^{\xi(\bm{t}-\bm{s},z)} \tau^{(a)}\left( \bm{t} -[z^{-1}] \right) \tau^{(a)}\left( \bm{s} + [z^{-1}]\right) = 0\; ,
\ee	
where \([z]=(z,z^2/2, z^3/3, \dots)\).

Consequently, the kernel
\be 
\label{eq:cond}
\boxed{
\mathcal{K}^{-1}_{\bm{u}}(\bm{v}) = \frac{\tau_{\bm{u}}^{(1)} (\bm{v})}{\tau_{\bm{u}}^{(2)} (\bm{v})} }
\ee
is the quotient of two tau functions.
\end{proof}

Assuming the risk of being too repetitive, we need to stress that if \(s>\sfrac{1}{2}\), the result above depends on a conjecture~\cite{Nepomechie:2016ejv, Nepomechie:2016ruv}.
\medskip

\begin{definition}
Furthermore, observe that \(G(\bm{u}, \bm{v}) \) does not have any interesting structure in \(\bm{v}\), since it is just the product of \(w(v_k)=v_k-v_k^{-1}\), therefore, one may define the normalized Slavnov product as
\be 
\bbracket{\Psi_{bethe}(\bm{u})}{\Psi(\bm{v})}\equiv \mathcal{K}^{-1}_{\bm{u}}(\bm{v})  = \frac{1}{G(\bm{u}, \bm{v})}\bracket{\Psi_{bethe}(\bm{u})}{\Psi(\bm{v})}\; .
\ee
\end{definition}
\medskip

Since the KP hierarchy is a particular case of the more general Toda lattice hierarchy, see~\cite{Zabrodin2018}, we can use those results in our case. From~\cite{Zinn-Justin:2002fzw, Zinn-Justin:2002rai, Zinnjustin2009}, we have the following lemma:
\begin{lemma}[Zinn-Justin-Zuber~\cite{Zinn-Justin:2002fzw, Zinn-Justin:2002rai, Zinnjustin2009}]
The functions~(\ref{eq:taufunc}) have integral representations
\be 
\label{eq:KPintrep}
\tau^{(a)}(\bm{t}) = (-)^{M(M-1)/2}  \det\limits_{i,j} \oint \frac{\dd z}{2\pi i} z^{-i} F^{(a)}_j (z) e^{\xi(\bm{t}, z^{-1})}\; ,
\ee
where \(t_m = \frac{1}{m}\sum_{i=1}^M v_i^m\) are the Miwa coordinates and
\be 
 \xi(\bm{t}, z^{-1}) = - \sum_{i=1}^M \ln (1-z^{-1} v_i) = \sum_{m=1}^\infty t_m z^{-m}\; .
\ee 
\end{lemma}

\begin{proof}
The proof is straightforward from direct calculations. Write~(\ref{eq:KPintrep}) as
\bse
\be 
\tau^{(a)}(\bm{v}) = (-)^{M(M-1)/2}  \det\limits_{i,j} \oint \frac{\dd z}{2\pi i} z^{-i}  F^{(a)}_j (z)  \prod_{k=1}^M\frac{z}{z - v_k} \; ,
\ee
integrating over the poles, we obtain
\be 
\begin{split}
\tau^{(a)}(\bm{v})  & = (-)^{M(M-1)/2} \det\limits_{i,j} \sum_{k=1}^M \frac{v_k^{M - i} }{ \prod\limits_{m\neq k} (v_k - v_m) } F^{(a)}_j(v_k)\\
& = (-)^{M(M-1)/2} \det\limits_{k,i} \left( \frac{v_k^{M-i} }{ \prod\limits_{m\neq v} (v_k - v_m) }\right) \det\limits_{k,j} F^{(a)}_j(v_k)\\ 
& = (-)^{M(M-1)/2} \frac{\Delta(\bm{v}) }{ \prod\limits_{ \substack{ k,m=1\\ m\neq k}}^M (v_k - v_m) }\ \det\limits_{i,j} F^{(a)}_j(v_i)
\end{split}
\ee
where in the first line, the summation denotes a matrices product. Moreover, using 
\[
\prod\limits_{ \substack{ k,m=1\\ m\neq k}}^M (v_k - v_m) = (-)^{M(M-1)/2}\Delta(\bm{v})^2\; \, 
\] 
we obtain
\[
\tau^{(a)}(\bm{v}) =
\frac{ \det\limits_{i,j} F^{(a)}_j(v_i) }{ \Delta(\bm{v})}\; ,
\]
that proves the equivalence of the two expressions. 
\ese
\end{proof}	
\medskip

\begin{corollary} The Slavnov product \(\bbracket{\Psi_{bethe}(\bm{u})}{\Psi(\bm{v})}\), as a function of the Miwa coordinates, can be represented as the normalized (matrix) expectation value
\bse
\be 
\bbracket{\Psi_{bethe}(\bm{u})}{\Psi(\bm{t})} = 
\frac{1}{Z(\bm{t})} \oint \prod_{j=1}^M \frac{\dd z_j}{2\pi i} e^{-M V(\bm{t},\bm{z})} \Delta^2(\bm{z}) \mathcal{K}^{-1}_{\bm{u}}(\bm{z}) \; ,
\ee
where the potential is given by
\be 
V(\bm{t},\bm{z}) = - \frac{1}{M} \sum_{j=1}^M \left( \xi(\bm{t}, z_j^{-1}) -M \ln z_j \right)  - \frac{1}{M} \ln \tau^{(2)}(\bm{z})\; .
\ee 
\ese
\end{corollary}

\begin{proof}
Using an integral representation~(\ref{eq:KPintrep}) and the Andreev identity
\be 
\det\limits_{ij} \left( \int \dd \mu(z) f_i(z) g_j(z) \right)=\frac{1}{M!}\int \prod_{l=1}^N \dd \mu(z_l) \det\limits_{jk}(f_{j}(z_k)) \det\limits_{jk}(g_{j}(z_k))\; ,
\ee 
we have 
\bse
\be 
\tau_{\bm{u}}^{(1)}(\bm{t}) = \frac{(-)^{M(M-1)/2}}{M!}
\oint \prod_{j=1}^M \frac{\dd z_j}{2\pi i} e^{\sum_{j=1}^M \left( \xi(\bm{t}, z_j^{-1}) -M \ln z_j \right) } \Delta(\bm{z})^2 \tau_{\bm{u}}^{(1)}(\bm{z})\; ,
\ee
where we use the measure \(\dd \mu(z) = e^{\xi(\bm{t}, z_j^{-1})}z^{-M}\dd z\). Expression~(\ref{eq:cond}) gives
\be 
\begin{split}
\tau^{(1)}_{\bm{u}}(\bm{t}) & = \frac{(-)^{M(M-1)/2}}{M!} \oint \prod_{j=1}^M \frac{\dd z_j}{2\pi i} e^{\sum_{j=1}^M \left( \xi(\bm{t}, z_j^{-1}) -M \ln z_j \right) } \Delta(\bm{z})^2 \mathcal{K}^{-1}_{\bm{u}}(\bm{z}) \tau_{\bm{u}}^{(2)}(\bm{z})\\
& = \frac{(-)^{M(M-1)/2}}{M!} \oint \prod_{j=1}^M \frac{\dd z_j}{2\pi i} e^{\sum_{j=1}^M \left( \xi(\bm{t}, z_j^{-1}) -M \ln v_j \right) + \ln \tau^{(2)}(\bm{z}) } \Delta(\bm{z})^2 \mathcal{K}^{-1}_{\bm{u}}(\bm{z})\\
& = \frac{(-)^{M(M-1)/2}}{M!} \oint \prod_{j=1}^M \frac{\dd z_j}{2\pi i} e^{ - M V(\bm{t}, \bm{z}) } \Delta(\bm{z})^2 \mathcal{K}^{-1}_{\bm{u}}(\bm{z})
\end{split}
\ee
where we have defined the effective potential \(V(\bm{t}, \bm{z})\) as
\be 
- M V(\bm{t}, \bm{z})=\sum_{j=1}^M \left( \xi(\bm{t}, z_j^{-1}) -M \ln z_j \right) + \ln \tau^{(2)}(\bm{z}) \; .
\ee

Defining the partition function \(Z(\bm{t})\) via
\be 
\frac{(-)^{M(M-1)/2}}{M!} Z(\bm{t})\equiv \tau^{(2)}(\bm{t}) = \oint \prod_{j=1}^M \frac{\dd z_j}{2\pi i} e^{ - M V(\bm{t}, \bm{z}) } \Delta(\bm{z})^2 \; ,
\ee
we see that the kernel \(\mathcal{K}^{-1}_{\bm{u}}(\bm{t})\), in Miwa coordinates, is given by
\be 
\bbracket{\Psi_{bethe}(\bm{u})}{\Psi(\bm{t})}\equiv \langle  \mathcal{K}^{-1}_{\bm{t}}(\bm{v})\rangle  = \frac{1}{Z(\bm{t})} \oint \prod_{j=1}^M \frac{\dd z_j}{2\pi i} e^{ - M V(\bm{t}, \bm{z}) } \Delta(\bm{z})^2 \mathcal{K}^{-1}_{\bm{u}}(\bm{z})\; .
\ee
\ese
\end{proof}

\subsection{Consistency conditions}

It is well known in the theory of classical integrable systems~\cite{Zabrodin2018} that the \ac{hbe}~(\ref{eq:hbe}) are equivalent to
\be 
\sum_{j=0}^\infty h_j(-2\bm{x}) h_{j+1}(\tilde{\bm{D}}) e^{\sum_{k=1}^\infty x_k D_k }  \llbracket \tau^{(a)}\left( \bm{t}\right), \tau^{(a)}\left( \bm{t}\right) \rrbracket \; ,
\ee
where \(h(\bm{x})\) are symmetric polynomials defined by 
\be 
\exp\left(\sum_{k\geq 0} x_k z^k\right) = \sum_{j\geq 0}h_{j}(\bm{x}) z^j\; .
\ee 
Moreover, 
\( \bm{D}\equiv \{D_k \ |\ k\in \mathbb{N}^\ast\}\) are the Hirota derivatives defined by the conditions
\be 
P(\bm{D}) \llbracket f(\bm{t}),g(\bm{t})\rrbracket := P(\partial_{\bm{x}}) \left( f(\bm{t}-\bm{x}) g(\bm{t}+ \bm{x}) \right)\big|_{\bm{x}=\bm{0}} = 0\; ,
\ee 
where \(P(\bm{x})\) is a generic polynomial, and the operators \( \tilde{\bm{D}}\) are Hirota derivatives defined with respect to \(\tilde{\partial}_{\bm{x}} = \{ k^{-1}\partial_{x_k}\ |\ k\in \mathbb{N}^\ast\} \). From these expressions, it is obvious that the kernel of the Slavnov product \(\mathcal{K}^{-1}_{\bm{u}}(\bm{t})\) satisfies the following hierarchy of differential equations
\be 
\label{eq:compcond}
\begin{split}
&\sum_{j=0}^\infty h_j(-2\bm{x}) h_{j+1}(\tilde{\bm{D}}) e^{\sum_{k=1}^\infty x_k D_k }  \llbracket \mathcal{K}^{-1}_{\bm{u}}(\bm{t}) \tau^{(2)}( \bm{t}), 
\mathcal{K}^{-1}_{\bm{u}}(\bm{t}) \tau_{\bm{u}}^{(2)}( \bm{t})\rrbracket = 0\\
&\sum_{j=0}^\infty h_j(-2\bm{x}) h_{j+1}(\tilde{\bm{D}}) e^{\sum_{k=1}^\infty x_k D_k }  \llbracket \tau_{\bm{u}}^{(1)}( \bm{t}), \mathcal{K}^{-1}_{\bm{u}}(\bm{t}) \tau^{(2)}( \bm{t})\rrbracket = 0\\
&\sum_{j=0}^\infty h_j(-2\bm{x}) h_{j+1}(\tilde{\bm{D}}) e^{\sum_{k=1}^\infty x_k D_k }  \llbracket \tau_{\bm{u}}^{(a)}( \bm{t}), \tau_{\bm{u}}^{(a)}( \bm{t}) \rrbracket = 0\; , \quad a=1,2
\end{split}	\; .
\ee 
The first two equations are, naturally, equivalent to the third one for \(a=2\), therefore the hierarchy of equations above is satisfied by the Slavnov product. 

\begin{remark}
It might be interesting to think of~(\ref{eq:compcond}) in the opposite direction: Assume the existence of a tau function \(\tau_{\bm{u}}^{(2)} (\bm{t})\), we would like to find a function \(\mathcal{K}_{\bm{t}}(\bm{t})\) such that \(\mathcal{K}^{-1}_{\bm{u}}(\bm{t})\tau_{\bm{u}}^{(2)} (\bm{t})\equiv \tau_{\bm{u}}^{(1)} (\bm{t})\) is another tau function. Solving  equation~(\ref{eq:compcond}), one might build such a function. 
\end{remark}

The first coefficients in \(x_1\) are
\be 
\begin{split}
D_1 & \llbracket f(\bm{t}), g(\bm{t})\rrbracket = 0\\ 
D_2& \llbracket f(\bm{t}), g(\bm{t})\rrbracket = 0\\
\left( D_1^3 - 4D_3 \right)& \llbracket f(\bm{t}), g(\bm{t})\rrbracket = 0\\
\left( D_1^4+ 3 D_2^2 - 4 D_1 D_3 + 3 D_1^2 D_2 - 6 D_4 \right)& \llbracket f(\bm{t}), g(\bm{t})\rrbracket = 0
\end{split}
\ee
As usual, when \(f(\bm{t})=g(\bm{t})\), as in the case of the KP tau functions, the first three equations are trivial, and the last is exactly the KP equation. For \(f(\bm{t})=\tau_{\bm{u}}^{(2)}(\bm{t})\) and \(g(\bm{t})=\mathcal{K}^{-1}_{\bm{u}}(\bm{t})\tau_{\bm{u}}^{(2)}(\bm{t})\), the first three equations define further conditions on \(\mathcal{K}^{-1}_{\bm{u}}(\bm{t})\).

\subsection{Pl\"{u}cker relation}

For the sake of completeness, we finish this section with a simple corollary that follows from the results in~\cite{Zinn-Justin:2002fzw, Zinn-Justin:2002rai, Zinnjustin2009} originally obtained for the Toda lattice. 
\begin{corollary2} 
For tau functions of the form~(\ref{eq:taufunc}), the \ac{hbe} are equivalent to the Pl\"{u}cker relation
\be 
\sum_{i=1}^{M+1} (-)^i \det\limits_{i,j}F^{(a)}_j\left(\bm{X}\setminus \{x_i\}\right) \det\limits_{i,j}F^{(a)}_j\left(\bm{Y}\cup \{x_i\}\right) = 0\; ,
\ee
where 
\be 
\label{eq:vec}
\begin{split}
\bm{X} & =(x_1, \dots, x_{M+1})\\
\bm{Y} & =(y_1, \dots, y_{M-1})
\end{split}\; .
\ee
\end{corollary2}

\begin{proof}
Let us consider the vectors~(\ref{eq:vec}), and the Miwa variables \(\bm{t}=\{t_p\}\) and \(\bm{s}=\{s_q\}\) with components
\be 
t_p = \frac{1}{p}\sum_{j=1}^{M+1} x_j^p\qquad 
s_q = \frac{1}{q}\sum_{j=1}^{M-1} y_j^q\; .
\ee

Let us start a proof with some identities.

\noindent {\bf \S.1} For a generic point \(i\), we have
\be 
\begin{split}
	\bm{t}-[x_i] & =\left\{ t_p - \frac{x_i^p}{p} \right\} = 
	\left\{ \frac{1}{p} \sum_{j=1}^{M+1} x_j^p - \frac{x_i^p}{p} \right\} = 
	\left\{ \frac{1}{p} \sum_{\substack{j=1\\ j\neq i}}^{M+1} x_j^p \right\}\\
	\bm{s}+[x_i] & =\left\{ s_p +\frac{a_i^p}{p} \right\} = 
	\left\{ \frac{1}{p} \sum_{j=1}^{M-1} y_j^p + \frac{x_i^p}{p} \right\} = 
	\left\{ \frac{1}{p} \sum_{j=1}^{M} y_j^p \right\}\;, \quad  \textrm{with}
	\quad y_{M}\equiv x_i
\end{split}\; .
\ee
Therefore
\be 
\begin{split}
	\tau^{(a)}(\bm{t}-[x_i]) & \equiv \tau^{(a)}(\bm{X}\setminus \{x_i\})\\
	\tau^{(a)}(\bm{s}+[x_i]) & \equiv \tau^{(a)}(\bm{Y}\cup\{x_i\})
\end{split}\; ,
\ee
where \(\card(\bm{X}\setminus \{x_i\})=\card(\bm{Y}\cup \{x_i\}) = M\).

\noindent {\bf \S.2} We want to write the Vandermonde determinant \(\Delta(\bm{X}\setminus \{x_i\})\) in terms of \(\Delta(\bm{X})\). For \(i\in [1, \dots, M+1]\), we have
\be 
\begin{split}
\Delta(\bm{X}) & = \prod_{\substack{j,k=1\\ j<k}}^{M+1} (x_j - x_k) =  \prod_{\substack{j=1}}^{i-1} (x_j - x_i) 
\prod_{\substack{k=i+1}}^{M+1} (x_i - x_k)  \prod_{\substack{j<k\\ j,k\neq i}}^{M+1} (x_j - x_k)\\
&= (-1)^{M+1-i} \Delta(\bm{X}\setminus \{x_i\}) \prod_{\substack{j=1\\ j\neq i}}^{M+1} (x_j - x_i) \; .
\end{split}	
\ee

\noindent {\bf \S.3} Similarly, we want to write the Vandermonde determinant \(\Delta(\bm{y}\cup \{x_i\})\) in terms of \(\Delta(\bm{y})\). For \(i\in [1, \dots, M+1]\), we write \(y_M\equiv x_i\); therefore
\be 
\begin{split}
\Delta(\bm{Y}\cup \{x_i\}) & = \prod_{\substack{j,k=1\\ j<k}}^{M} (y_j - y_k)= \prod_{j=1}^{M-1} (y_j - x_i)\prod_{\substack{j,k=1\\ j<k}}^{M-1} (y_j - y_k)\\
& = \Delta(\bm{Y}) \prod_{j=1}^{M-1} (y_j - x_i)\; .
\end{split}	
\ee

\noindent {\bf \S.4} With these definitions, it is easy to show that
\be 
e^{\xi(\bm{t}-\bm{s} ,z)} = \prod_{j=1}^{M+1}\frac{1}{1-z x_j}  \prod_{k=1}^{M-1}(1-z y_k)\; ,
\ee
and inserting it back into the \ac{hbe}~(\ref{eq:hbe}), we have
\be 
\begin{split}
0  & = \oint \frac{\dd z}{2\pi i}  \prod_{j=1}^{M+1}\frac{1}{1-z x_j}  \prod_{k=1}^{M-1}(1-z y_k) \tau^{(a)}\left( \bm{t} -[z^{-1}] \right) \tau^{(a)}\left( \bm{s} + [z^{-1}]\right)\\
& = \oint \frac{\dd z}{2\pi i} \frac{(-x_i^{-1})}{ z- x_i^{-1}} \prod_{\substack{j=1\\ j\neq i} }^{M+1}\frac{1}{1-z x_j}  \prod_{k=1}^{M-1}(1-z y_k) \tau^{(a)}\left( \bm{t} -[z^{-1}] \right) \tau^{(a)}\left( \bm{s} + [z^{-1}]\right)\\ 
& = \sum_{i=1}^{M+1} (-)  \prod_{\substack{j=1\\ j\neq i} }^{M+1}\frac{1}{x_i - x_j}  \prod_{k=1}^{M-1}(x_i - y_k) \tau^{(a)}\left( \bm{t} -[x_i ] \right) \tau^{(a)}\left( \bm{s} + [x_i ]\right)\; .
\end{split}
\ee	

Using the expressions derived in~{\bf \S.1},~{\bf \S.2} and~{\bf \S.3}, we find
\be 
0 = \sum_{i=1}^{M+1} (-)^{M+i-1}\frac{\Delta\left(\bm{X}\setminus \{x_i\}\right)}{\Delta\left(\bm{X}\right)} 
\frac{\Delta\left(\bm{Y}\cup \{x_i\}\right)}{\Delta\left(\bm{Y}\right)}  \tau\left(\bm{X}\setminus \{x_i\}\right)
\tau\left(\bm{Y}\cup \{x_i\}\right)
\ee
and from the form~(\ref{eq:taufunc}) for the tau functions, we conclude that the \ac{hbe} become the Pl\"{u}cker relation
\be 
\sum_{i=1}^{M+1} (-)^i \det\limits_{i,j}F^{(a)}_j\left(\bm{X}\setminus \{x_i\}\right) \det\limits_{i,j} F^{(a)}_j\left(\bm{Y}\cup \{x_i\}\right) = 0\; .
\ee
Observe that although we have used the indices \(a\) to keep the notation consistent in this work, this identity is a general statement. Moreover, observe that in~\cite{Zinn-Justin:2002fzw, Zinn-Justin:2002rai, Zinnjustin2009} the formula for the Toda lattice (that can be obtained using the same strategy discussed above) does not have the factors \((-)^i\), and it spoils the relations (unless, of course, these signs are implicit).
\end{proof}	

%%%%%%%%%%%%%%%%%%%%%%%%%%%%%%%%%%%%%%%%%%%%%%
%%%%%%%%%%%%%%%%%%%%%%%%%%%%%%%%%%%%%%%%%%%%%%
%%%%%%%%%%%%%%%%%%%%%%%%%%%%%%%%%%%%%%%%%%%%%%
%%%%%%%%%%%%%%%%%%%%%%%%%%%%%%%%%%%%%%%%%%%%%%

\section{Schur functions expansion \& Baker-Akhiezer function}
\label{sec:schur}

This section discusses some properties of the Schur polynomials expansion of the tau-functions and Slavnov product  described above. Let us summarize a few facts about the Schur polynomials expansion of the tau function~(\ref{eq:fftau}). Here we use the notation~\cite{Alexandrov:2012tr}, see also~\cite{Miwa2000},

Tau functions of the KP hierarchy, with Miwa coordinates \(\bm{t}=\{t_1, t_2, t_3, \dots\}\), can be represented in terms of free fermions as
\be 
\label{eq:fftau}
\tau_{\bm{u}}^{(a)}(\bm{t})=\langle 0 | e^{\bm{J}_+(\bm{t})} G^{(a)}| 0\rangle \; ,
\ee 
where \(\bm{J}_+(\bm{t})=\sum_{k>0} t_k J_k\), where the components \(\{J_k\ |\ k \in \mathbb{Z}\}\) satisfy the Heisenberg algebra \(\mathfrak{\hat{u}}(1)\) 
\[ 
[J_m, J_n]=m \delta_{m+n,0}\; .
\] 
Moreover, we also need the group element \(G\in GL(\infty)\). Therefore
\be 
\begin{split}
	\tau_{\bm{u}}^{(a)}(\bm{t}) & =\langle 0 | e^{\bm{J}_+(\bm{t})} G^{(a)}| 0\rangle = \sum_{\lambda} (-)^{\flat(\lambda)} \langle 0 | e^{\bm{J}_+(\bm{t})} |\lambda\rangle \langle \lambda| G^{(a)}| 0\rangle\\
	& = \sum_{\lambda}  (-)^{\flat(\lambda)} \langle \lambda| G^{(a)}| 0\rangle s_{\lambda}(\bm{t})\; ,
\end{split}
\ee
where the integer partitions in Frobenius coordinates are denoted as \(\lambda=(\alpha_1,\dots, \alpha_d |\beta_1, \dots, \beta_{d})\), and we have used that  \(\flat(\lambda)=\sum_{j=1}^d (\beta_j + 1)\).

Now, it is easy to pass from Miwa coordinates \(\bm{t}\) to the \emph{configuration} space \(\bm{v}\)
\be 
\tau_{\bm{u}}^{(a)}(\bm{v}) = \sum_{\lambda} (-)^{\flat(\lambda)} \langle \lambda| G^{(a)} |0\rangle s_{\lambda}(\bm{v})\; ,
\ee
and using the orthonormality of Schur polynomials~\cite{Orlov:2002jx}, we have 
\be 
\langle \lambda| G^{(a)} |0\rangle = (-)^{\flat(\lambda)} (s_{\lambda}(\bm{v}), \tau_{\bm{u}}^{(a)}(\bm{v}) )\; .
\ee
\medskip

\begin{remark}
Observe that although the Miwa coordinates \(\bm{t}\) have infinitely many components \(\{t_k\ |\ k\in \mathbb{N}^\ast\}\), the vectors \(\bm{v}=\{v_1, v_2, \dots, v_M\}\) are limited by the dimension of the Magnon sector \(M\). Therefore, the coordinates \(\{t_m\ |\  m>M\}\), do not carry any new information. It also means that it would be enough to consider \(G\in GL(M, \infty)\)~\cite{Miwa2000, Alexandrov:2012tr}.
\end{remark}

\subsection{Slavnov product in terms of Schur functions}

We want to write the series expansion for the functions \(F_j^{(a)}(z)\). From the eigenvalue expression~(\ref{eq:eigenvalues}), it is easy to show that
\be 
\Lambda(z, \bm{u}) = - \frac{1+q^{2(N-2M+1)}}{ q^{2(N-M)+1} z^{2N}} + \sum_{n=-N+1}^\infty \tilde{f}_{2n} z^{2n}\; ,
\ee
where the coefficients \(\tilde{f}_{2j}\) depend on the Bethe roots \(\bm{u}\). Therefore,
\be 
F^{(1)}_j(z) = \partial_{u_j}\Lambda(z, \bm{u})= \sum_{n=-N+1}^\infty \tilde{f}^{(1)}_{j, 2n} z^{2n}
\ee
where \(\tilde{f}^{(1)}_{j, 2n}\equiv \partial_{u_j} \tilde{f}^{(1)}_{2n}\). Let us write these functions as 
\be 
\label{eq:seriesF1}
F_j^{(1)}(z) = \frac{1}{z^{2N-2}} \sum_{n=0}^{\infty} f^{(1)}_{j, 2n}\ z^{2n}\; , \quad f^{(1)}_{j, 2n}\equiv \tilde{f}^{(1)}_{j, 2(n-N+1)}\; .
\ee

Additionally, the function \(F_j^{(2)}(z) \) has an expansion of the form 
\be 
\label{eq:seriesF2}
F_j^{(2)}(z) = z^2\sum_{n=0}^{\infty} f^{(2)}_{j, 2n}\ z^{2n}\; .
\ee

At this point we can simply consider the transformation \(z\to \pm \sqrt{y}\) and write the expansions above in terms of \(y\). The ambiguity in the definition of \(z\) is harmless since the function \(\Lambda(z)\) is even; this means that the physical results are the same regardless of the square root branch we use. Therefore
\bse 
\be 
F_j^{(1)}(y) = \frac{1}{y^{N-1}} \sum_{n=0}^{\infty} \hat{f}^{(1)}_{j, n}\ y^{n} \; , \quad \hat{f}^{(1)}_{j, 2n}\equiv f^{(1)}_{j, 2n} \; ,
\ee
and 
\be 
F_j^{(2)}(y) = y \sum_{n=0}^{\infty} \hat{f}^{(2)}_{j, n}\ y^{n} \; , \quad \hat{f}^{(2)}_{j, n}\equiv f^{(2)}_{j, 2n}\; .
\ee
\ese

As a simple consequence of the calculations above, the definition \(\bm{v}=\{\sqrt{w_i}\ |\ i=1, \dots, M\}\) yields
\bse
\be 
\det\limits_{ij} F_j^{(a)}(w_i) = \mathcal{P}^{(a)}(\bm{w}) \det\limits_{ij} \sum_{n=0}^{\infty} \hat{f}^{(a)}_{j, 2n}\ w_i^{2n}\; ,
\ee
where 
\be 
\mathcal{P}^{(a)}(\bm{w}) = \left\{ 
\begin{array}{lll}
	\prod_{j=1}^M w_j^{1-N} & \text{if} &	a=1\\ & \vspace{-7pt} & \\
	\prod_{j=1}^M w_j & \text{if} &	a=2
\end{array}
\right.\; .
\ee
\ese

From these expressions, one can find the Schur polynomials expansion using the standard algorithm~\cite{Foda2009a, Foda2009}, see also~\cite{Takasaki:2010qm}. The Cauchy-Binet theorem gives
\be 
\det\limits_{ij} \sum_{n=0}^{\infty} \hat{f}^{(a)}_{j, n}\ w_i^{n} = \sum_{0\leq \ell_M <\dots \ell_1\leq \infty} \det\limits_{i j} (f^{(a)}_{i , \ell_j}) \det\limits_{i j} (v_i^{ \ell_j})\; .
\ee
Now, we use the transformation \(\ell_j = \lambda_j - j +M\), then
\be 
\det\limits_{ij} \sum_{n=0}^{\infty} \hat{f}^{(a)}_{j, 2n}\ w_i^{2n} = \sum_{0\leq \lambda_M \leq \dots \leq \lambda_1\leq \infty} \det\limits_{i j} (\hat{f}^{(a)}_{i , \lambda_j - j +M}) \det\limits_{i j} (w_i^{\lambda_j - j +M})\; ,
\ee
where the sum is over the Young diagrams \(\bm{\lambda}=(\lambda_1, \dots, \lambda_M)\). Here, \(c^{(a)}_\lambda(\bm{u}) \equiv \det(f^{(a)}_{i, \lambda_j - j +M})\) are Pl\"{u}cker coordinates in Sato Grassmannian. Finally, one concludes that
\be 
\tau^{(a)}_{\bm{u}}(\bm{w}) =\frac{\det\limits_{ij} F_j^{(a)}(w_i)}{\Delta(\bm{w})} = \mathcal{P}^{(a)}(\bm{w}) 
\sum_{ \bm{\lambda}} c^{(a)}_{\bm{\lambda}}(\bm{u}) s_{\bm{\lambda}}(\bm{w})\; .
\ee

All in all, the normalized Slavnov product becomes
\bse
\be 
\label{eq:nslav}
\bbracket{\Psi_{bethe}(\bm{u})}{\Psi(\bm{w})} = \prod_{j=1}^{M}\frac{1}{w_j^N} \frac{\sum_{ \bm{\lambda}} c^{(1)}_{\bm{\lambda}}(\bm{u}) s_{\bm{\lambda}}(\bm{w})}{\sum_{ \bm{\lambda}} c^{(2)}_{\bm{\lambda}}(\bm{u}) s_{\bm{\lambda}}(\bm{w})}\; .
\ee
Writing this expression in terms of a new set of Miwa coordinates \(t'_m = \frac{1}{m}\sum_{j=1}^M w_j^m\), we find
\be 
\label{eq:nslav2}
\bbracket{\Psi_{bethe}(\bm{u})}{\Psi(\bm{t}')} = \prod_{j=1}^{M}\frac{1}{w_j^N(\bm{t}')} \frac{\sum_{ \bm{\lambda}} c^{(1)}_{\bm{\lambda}}(\bm{u}) s_{\bm{\lambda}}(\bm{t}')}{\sum_{ \bm{\lambda}} c^{(2)}_{\bm{\lambda}}(\bm{u}) s_{\bm{\lambda}}(\bm{t}')}\; .
\ee
\ese
Since there is no ambiguity in our discussion, in the remainder of this work we drop the primes in \(t'_m\). Using some elementary properties of the Schur polynomials~\cite{Macdonald1998}, it is easy to see that
\bse 
\be 
\begin{split}
\sum_{ \bm{\lambda}} c^{(2)}_{\bm{\lambda}}(\bm{u}) s_{\bm{\lambda}}(\bm{t}) & = 
c^{(2)}_{\bm{\emptyset}}(\bm{u}) + \sum_{ \bm{\lambda}>\bm{\emptyset}} c^{(2)}_{\bm{\lambda}}(\bm{u}) s_{\bm{\lambda}}(\bm{t})\\
& = c^{(2)}_{\bm{\emptyset}}(\bm{u}) \left(1 + \frac{1}{c^{(2)}_{\bm{\emptyset}}(\bm{u})}\sum_{ \bm{\lambda}>\bm{\emptyset}} c^{(2)}_{\bm{\lambda}}(\bm{u}) s_{\bm{\lambda}}(\bm{t})\right)
\end{split}
\ee
where we have used that \(s_{\bm{\emptyset}}(\bm{t})=1\). Then, we have the formal series
\be 
\sum_{ \bm{\lambda}} c^{(2)}_{\bm{\lambda}}(\bm{u}) s_{\bm{\lambda}}(\bm{t})
 \equiv c^{(2)}_{\bm{\emptyset}}(\bm{u}) \left( 1 - \sum_{ \bm{k}} \bar{c}_{\bm{k}}(\bm{u}) \vec{T}^{\bm{k}}\right)\; ,
\ee
where \(\vec{T}^{\bm{k}}\) are monomials \(t_1^{k_1} t_2^{k_2} t_3^{k_3}\cdots \) and the sum is over all possible (infinite) vectors \(\bm{k}=(k_1, k_2, \dots, k_j, 0,0, 0, \dots )\), \(k_i\in \mathbb{N}\) \(\forall \ i \in \mathbb{N}\) and \(j<\infty\). Additionally, the formal series
\be 
\label{eq:inv}
\frac{1}{\sum_{ \bm{\lambda}} c^{(2)}_{\bm{\lambda}}(\bm{u}) s_{\bm{\lambda}}(\bm{t})}
= \sum_{ \bm{k}} \hat{c}_{\bm{k}}(\bm{u}) \vec{T}^{\bm{k}}\; .
\ee
\ese
is a symmetric function in \(\bm{t}\), since the \(\{s_{\bm{\lambda}}(\bm{t})\}\) form a basis for the space of symmetric functions. It also means that function~(\ref{eq:inv}) can also be expanded in terms of Schur functions. Finally, inserting~(\ref{eq:inv}) into~(\ref{eq:nslav2}) and using the Littlewood–Richardson rules \(s_{\bm{\mu}}(\bm{t}) s_{\bm{\nu}}(\bm{t}) = \sum_{\bm{\lambda}} N^{\bm{\lambda}}_{\bm{\mu} \bm{\nu} }s_{\bm{\lambda}}(\bm{t})\), one can see that 
\be 
\label{eq:nslav3}
\boxed{ \bbracket{\Psi_{bethe}(\bm{u})}{\Psi(\bm{t})} = \prod_{j=1}^{M}\frac{1}{w_j^N(\bm{t})} \sum_{ \bm{\lambda}} A_{\bm{\lambda}}(\bm{u}) s_{\bm{\lambda}}(\bm{w}) }
\ee
where the coefficients \(\{A_{\bm{\lambda}}(\bm{u})\} \) are written in terms of \(\{c^{(1)}_{\bm{\lambda}}(\bm{u})\}\), \(\{c_{\bm{\lambda}}^{(2)}(\bm{u})\}\) and of the Littlewood–Richardson coefficients \(\{ N^{\bm{\lambda}}_{\bm{\mu} \bm{\nu}} \}\). Observe that the expansion in~(\ref{eq:nslav3}) does not define a tau function, since the coefficients \(\{A_{\bm{\lambda}}(\bm{u})\}\) are not Pl\"ucker coordinates. 

\subsection{Baker-Akhiezer functions}

We already know that the series expansions \(\sum_{ \bm{\lambda}} c^{(a)}_{\bm{\lambda}}(\bm{u})s_{\bm{\lambda}}(\bm{w})\) are, by themselves, tau functions
\be 
\tilde{\tau}_{\bm{u}}^{(a)}(\bm{v}) = \sum_{ \bm{\lambda}} c^{(a)}_{\bm{\lambda}}(\bm{u}) s_{\bm{\lambda}}(\bm{w})\; .
\ee
Additionally, 
\be 
\prod_{j=1}^{M}\frac{1}{w_j^N} = \exp \left( - N \sum_{p=1}^\infty\sum_{j=1}^M\frac{1}{p}(1-w_j)^p \right)\; .
\ee
These results imply that the normalized Slavnov product~(\ref{eq:nslav}) can be represented as
\be 
\bbracket{\Psi_{bethe}(\bm{u})}{\Psi(\bm{w})} = 
\exp \left( - N \sum_{p=1}^\infty\sum_{j=1}^M\frac{1}{p}(1-w_j)^p \right) \frac{\tilde{\tau}_{\bm{u}}^{(1)}(\bm{w})}{\tilde{\tau}_{\bm{u}}^{(2)}(\bm{w})}\; ,
\ee
that resembles the Baker-Akhiezer functions for the KP hierarchy. Inspired by this similarity, one may fell compelled to define the matrix \(\Psi(\bm{t})\) in Miwa coordinates as
\be 
\Psi(\bm{t},z) = e^{\xi(\bm{t},z)}
\begin{pmatrix}
	\psi_{11}(\bm{t},z) & \psi_{12}(\bm{t},z) \\ \psi_{21}(\bm{t},z) & \psi_{22}(\bm{t},z)
\end{pmatrix}\;  , \quad 
\psi_{ab}(\bm{t},z) = \frac{\tilde{\tau}_{\bm{u}}^{(a)}(\bm{t} - [z^{-1}])}{\tilde{\tau}_{\bm{u}}^{(b)}(\bm{t})} \; , \quad a,b=1,2\; .
\ee
The diagonal elements are Baker-Akhiezer tau functions, while anti-diagonal terms give the Slavnov product 
\be 
\bbracket{\Psi_{bethe}(\bm{u})}{\Psi(\bm{t})} = \prod_{j=1}^{M}\frac{1}{w_j^N(\bm{t})} \psi_{12}(\bm{t},\infty)\; ,
\ee
where \(\psi_{12}(\bm{t},\infty)=\psi_{21}(\bm{t},\infty)^{-1} = \tilde{\tau}_{\bm{u}}^{(1)}(\bm{t})/ \tilde{\tau}_{\bm{u}}^{(2)}(\bm{t}) \). Obviously, this latter condition does not necessarily mean that \(\psi_{12}(\bm{t},z)=\psi_{21}(\bm{t},z)^{-1}\). 

We can study the compatibility conditions for the linear problem
\be 
D_m \Psi(\bm{t},z) = \mathcal{A}_m \Psi(\bm{t},z) 
\ee
where \(D_m\) are differential operators depending on the Miwa coordinates \(\bm{t}\). The simplest choice for the operator \(D_m\) is \(D_m = \mathbb{1}\partial_m\), but this condition  means that \(\Psi(\bm{t},z)\) is itself a \emph{Baker-Akhiezer} function. Obviously, it imposes further undesirable constrains on \(\psi_{21}(\bm{t},z)\). Another possibility is the ansatz 
\be 
D_m =
\begin{pmatrix}
\partial_m & \mathcal{M}_m \\
\mathcal{M}_m & \partial_m
\end{pmatrix} 
\equiv \mathbb{1}\partial_m +\sigma^1 \mathcal{M}_m\; .
\ee
From \(D_m D_n \Psi(\bm{t},z)  = D_n D_m \Psi(\bm{t},z)\) and the zero curvature equation for the KP hierarchy, it is easy to write the conditions for \(\mathcal{M}_m\) and \(\mathcal{A}_m\). The consequences of these definitions, on the other hand, are not clear yet. This structure is currently under investigation, and we expect to report some results in a separate work. 

%%%%%%%%%%%%%%%%%%%%%%%%%%%%%%%%%%%%%%%%%%%%%%
%%%%%%%%%%%%%%%%%%%%%%%%%%%%%%%%%%%%%%%%%%%%%%
%%%%%%%%%%%%%%%%%%%%%%%%%%%%%%%%%%%%%%%%%%%%%%
%%%%%%%%%%%%%%%%%%%%%%%%%%%%%%%%%%%%%%%%%%%%%%

\section{Final remarks and conclusions}
\label{sec:remarks}

We have studied the relations between open Temperley-Lieb spin chains and tau functions of the KP hierarchy. This analysis is a logical continuation of the works of Foda, Wheeler and Zuparic~\cite{Foda2009a, Foda2009} for the spin-\(\sfrac{1}{2}\)-XXZ spin chain with periodic boundary conditions. In this brief section we would like to  summarize the main results we have obtained and compare them with~\cite{Foda2009a, Foda2009}. In the sequence, we discuss some open problems.  

First of all, let us briefly review the work~\cite{Foda2009a, Foda2009}, using -- in part -- the language discussed above. Kitanine, Maillet and Terras~\cite{Kitanine1999} have shown that given the Bethe roots \(\bm{\lambda}=\{\lambda_i \ | \ i=1, \dots, M\}\), and a set of free variables \(\bm{\mu}=\{\mu_i \ | \ i=1, \dots, M\}\), the Slavnov products~\cite{Slavnov1989} of the XXZ spin chain with periodic boundary conditions can be written as
\be 
\bracket{\Phi(\bm{\mu})}{\Phi_{bethe}(\bm{\lambda})} = \frac{\det \bm{T}}{\det \bm{V}}
\ee
where \(\bm{T}\) and \(\bm{V}\) are matrices with components
\be 
T_{ij}\equiv \mathcal{T}_i (\mu_j,\bm{\lambda}) = \frac{\partial}{\partial \lambda_i}\Lambda(\mu_j, \bm{\lambda})
\quad\text{and} \quad 
V_{ij}=\frac{1}{\sinh(\mu_j - \lambda_i)}
\ee
respectively. Additionally, \(\Lambda(\mu , \bm{\lambda})\) denotes the eigenvalues of the transfer matrix  -- whose explicit form can be found in~\cite{Kitanine1999}.

The change of coordinates \(v_i = e^{\mu_i/2}\) and \(u_i = e^{\lambda_i/2}\) yields
\be 
T_{ij} \equiv \mathcal{T}_i (v_j,\bm{u}) = \frac{1}{2}\frac{\partial}{\partial \log u_i}\Lambda(v_j, \bm{u})
\quad\text{and} \quad 
V_{ij}=\frac{2 u_i^{1/2} v_j^{1/2}}{v_j - u_i}\; .
\ee
Finally, using the Cauchy's determinant 
\be 
\det\limits_{i,j}\left( \frac{1}{u_i - v_j}\right) = \frac{\Delta(\bm{u}) \Delta(\bm{v})}{\prod\limits_{i,j}(u_i - v_j)}\; ,
\ee
we immediately see that the Slavnov product is proportional to a KP tau function
\be 
\bracket{\Phi(\bm{v})}{\Phi_{bethe}(\bm{u})} = G(\bm{v}, \bm{u}) \frac{\det\limits_{i,j} \mathcal{T}_i (v_j,\bm{u}) }{\Delta(\bm{v})}\; ,
\ee
where we have collected the unimportant multiplicative terms in the function \(G(\bm{v}, \bm{u})\). Finally, using the ideas reviewed in~\cite{Takasaki:2010qm} -- and in our previous section --, one can recover the Schur polynomials expansion of~\cite{Foda2009a, Foda2009}. 

Moving to the results discussed in this paper, it well known that the Hamiltonian~(\ref{eq:tlhamil}) for the spin-\(\sfrac{1}{2}\) is equivalent to the XXZ spin chain with \(U_q(\mathfrak{sl}(2))\) invariant boundary conditions~\cite{Pasquier:1989kd}, see~\cite{Doikou:2009xq} and references therein for a review on this topic. Theorem~(\ref{eq:sptheorem}) states that the Slavnov product for \ac{tl} open spin chain~(\ref{eq:tlhamil}) is proportional to a quotient of KP hierarchy tau functions, and as we have repeatedly stressed, conjecture~\cite{Nepomechie:2016ruv, Nepomechie:2016ejv} implies that this result is also true for spins \(s>\sfrac{1}{2}\). All these results can be summarized in table~\ref{table:maps}. 

\begin{table}[h!]
	\centering
	\begin{tabular}{|l|l|l|}
		\hline 
		\texttt{Spin chain} & \texttt{Slavnov product} / \(\tau\)-\texttt{functions} & \texttt{Ref.} \\ 
		\hline\hline
		XXZ - periodic \(s=\sfrac{1}{2}\) & \(\propto \tau(\bm{t})\)  & \textbf{FWZ}:~\cite{Foda2009a, Foda2009}\\ \hline
		XXZ - \(U_q(\mathfrak{sl}(2))\) \(s=\sfrac{1}{2}\) &  \(\propto\tau^{(1)}(\bm{t})/\tau^{(2)}(\bm{t})\)  & \textbf{Here}: Theorem~(\ref{eq:sptheorem})\\ \hline
		Temperley-Lieb \(s>\sfrac{1}{2}\) &  \(\propto\tau^{(1)}(\bm{t})/\tau^{(2)}(\bm{t})\)  & \textbf{Here}: Conjecture\\ \hline
	\end{tabular}	
	\caption{Relation between the type of spin chain and the KP tau functions}
	\label{table:maps}
\end{table}

There are many questions left unanswered, starting with the list given in the conclusions of~\cite{Foda2009a, Foda2009}. Those problems have not been properly addressed, and their study might lead to new breakthroughs in the field of integrable systems. One particularly interesting question is to understand how we can use this classical/quantum relation more pragmatically, since up to this point, we have just the description of a duality. There are also different connections between quantum spin chains and KP tau functions~\cite{Alexandrov:2011aa}, and it is not clear if there is any relation between the latter results and the structures we studied here.

There are also some immediate problems that follow from our analysis. First of all, the conjecture of~\cite{Nepomechie:2016ruv, Nepomechie:2016ejv} needs to be completely proved (or disproved) to settle down part of the questions raised in this work. Additionally, it is interesting to address this classical/quantum relation for other spin chains, representations and other boundary conditions. 

Two problems currently under investigation are the following. In section~\ref{sec:slavnov} we briefly described how the Slavnov product appears as a expectation value of a matrix model. Furthermore, we have also seen in section~\ref{sec:schur} that we can collect the same Slavnov product in terms of a matrix whose diagonal components are Baker-Akhiezer functions. We have not discussed any consequence of these two features, and hopefully we will be able to report some results in the future. Let us hope that nature does not disappoint us.

%%%%%%%%%%%%%%%%%%%%%%%%%%%%%%%%%%%%%%%%%%%%%%
%%%%%%%%%%%%%%%%%%%%%%%%%%%%%%%%%%%%%%%%%%%%%%
%%%%%%%%%%%%%%%%%%%%%%%%%%%%%%%%%%%%%%%%%%%%%%
%%%%%%%%%%%%%%%%%%%%%%%%%%%%%%%%%%%%%%%%%%%%%%

\subsection*{Acknowledgments.} This work was supported by the Swiss National Science Foundation under grant number \textsc{PP00P2\_183718/1}. I would like to thank Susanne Reffert and Domenico Orlando for many discussions and collaborations on related projects. 

\appendix

\section{Strict Young diagrams}

In section~\ref{sec:schur} we have considered a transformation \(z\to \sqrt{w}\) to simplify the sums~(\ref{eq:seriesF1}) and~(\ref{eq:seriesF2}). In this appendix, we tried to answer if the expansions in even powers could lead to new insights. Although the answer is negative, we show that this expansion gives, at least, some interesting combinatorial results. This section may be regarded as a mathematical curiosity.  

We discuss some simple consequences of the Schur polynomials expansions of the form
\bse
\be 
\det\limits_{ij} F_j^{(a)}(v_i) = \mathcal{P}^{(a)}(\bm{v}) \det\limits_{ij} \sum_{n=0}^{\infty} f^{(a)}_{j, 2n}\ v_i^{2n}\; ,
\ee
where 
\be 
\mathcal{P}^{(a)}(\bm{v}) = \left\{ 
\begin{array}{lll}
	\prod_{j=1}^M v_j^{2-2N} & \text{if} &	a=1\\ & \vspace{-7pt} & \\
	\prod_{j=1}^M v_j^{2} & \text{if} &	a=2
\end{array}
\right.\; .
\ee
\ese
The Cauchy-Binet theorem gives
\be 
\det\limits_{ij} \sum_{n=0}^{\infty} f^{(a)}_{j, 2n}\ v_i^{2n} = \sum_{\substack{0\leq \ell_M <\dots \ell_1\leq \infty\\ \ell_j \in 2\mathbb{N}\; , \forall j\in[1,M]}} \det\limits_{i j} (f^{(a)}_{i , \ell_j}) \det\limits_{i j} (v_i^{ \ell_j})\; .
\ee
Now, we use the transformation \(\ell_j = \lambda_j - j +M \in 2\mathbb{N}\), then
\be 
\det\limits_{ij} \sum_{n=0}^{\infty} f^{(a)}_{j, 2n}\ v_i^{2n} = \sum_{\substack{0\leq \lambda_M <\dots \lambda_1\leq \infty\\ \bm{\lambda}\in \mathcal{C}}} \det\limits_{i j} (f^{(a)}_{i , \lambda_j - j +M}) \det\limits_{i j} (v_i^{\lambda_j - j +M})\; ,
\ee
where \(\mathcal{C}\) is the set of Young diagrams \(\bm{\lambda}=(\lambda_1, \dots, \lambda_M)\) satisfying the constraints \( \lambda_j - j +M \in 2\mathbb{N}\). Here,  \(c^{(a)}_\lambda(\bm{u}) \equiv \det(f^{(a)}_{i, \lambda_j - j +M})\) are Pl\"{u}cker coordinates in Sato Grassmannian. Therefore, we conclude that
\be 
\tau^{(a)}_{\bm{u}}(\bm{v})=\frac{\det\limits_{ij} F_j^{(a)}(v_i)}{\Delta(\bm{v})} = \mathcal{P}^{(a)}(\bm{v}) 
\sum_{ \bm{\lambda}\in \mathcal{C}} c^{(a)}_\lambda(\bm{u}) s_{\bm{\lambda}}(\bm{v})\; .
\ee

Observe that the constraints in \(\mathcal{C}\) give interesting classes of diagrams. In particular, let us denote \(\mathbb{Z}_2=\{ [0],[1]\}\), then 
\be 
\label{eq:cond1}
\begin{split}
	M\in [0] & \Rightarrow \left\{ \begin{array}{lll} \lambda_{j} \in [0] & \textrm{if} & j\in[0]\\ \lambda_{j} \in [1] & \textrm{if} & j\in[1] \end{array} \right. \\
	M\in [1] & \Rightarrow \left\{ \begin{array}{lll} \lambda_{j} \in [1] & \textrm{if} & j\in[0]\\ \lambda_{j} \in [0] & \textrm{if} & j\in[1] \end{array} \right.
\end{split}\; .
\ee 
These conditions also mean that the tau functions \(\tau_{\bm{u}}^{(1)}(\bm{v})\) and \(\tau_{\bm{u}}^{(2)}(\bm{v})\) are summed over \emph{strict} Young diagrams
\be
\label{eq:cond2}
\bm{\lambda} = (\lambda_1, \lambda_2, \dots, \lambda_M)\quad \text{with} \quad \lambda_1 > \lambda_2 > \dots> \lambda_M\geq 0 \; .
\ee
\begin{remark}
	From these conditions, we conclude that the when written in terms of the coordinates \(\bm{v}\), the Slavnov product for the spin-\(s\) Temperley-Lieb algebra is the quotient of two tau functions with an expansion over Schur polynomials satisfying the conditions~(\ref{eq:cond1}) and~(\ref{eq:cond2}).
\end{remark}

\noindent We can study some particular cases:

\paragraph{\(\bullet\ \bm{M=1}\):} The vacuum Young diagram \(\bm{\emptyset}\) is admissible if and only if \(M=1\). In fact, for this case, the tau functions \(\tau_{\bm{u}}^{(a)}(\bm{t})\) have 1-row Young diagrams with an even number boxes, that is
\begin{equation}
	\ytableausetup{centertableaux, smalltableaux} \quad 
	\tau_{\bm{u}}^{(a)}(\bm{t})\propto \ c^{(a)}_{\bm{\emptyset}} \bm{\emptyset} + c^{(a)}_{(2)} \ydiagram[*(lightgray)]{2} + c^{(a)}_{(4)}\ydiagram[*(lightgray)]{4} + \cdots 
\end{equation}
where we have represented the Schur function by its corresponding Young diagram. 

Assuming that there exists a number \(\lambda_1^{max}\in 2\mathbb{N}\) such that \(c^{(a)}_{(\lambda_1)}\approx 0\)\ \(\forall \ \lambda_1 > \lambda_1^{max}\). Then, there will be \((\lambda_1^{max}+2)/2\) nonzero components \(c^{(a)}_{(\lambda_1)}\) (Young diagrams) contributing to the tau function \(\tau_{\bm{u}}^{(a)}(\bm{t})\). Evidently, there is, \emph{a priori}, no reason to believe that such a maximal value exists; and this comment is just a simple exercise. 
\medskip

\paragraph{\(\bullet\ \bm{M=2}\):} Now we have diagrams satisfying the conditions 
\be 
\bm{\lambda} = (\lambda_1, \lambda_2)\qquad \text{with} \quad \lambda_1>\lambda_2\; , \quad \lambda_1 \in 2\mathbb{N} + 1\; , \quad \lambda_2 \in 2\mathbb{N}\; .
\ee
The minimal diagram in this configuration is \(\lambda^{min}=(1,0)\). All in all, the tau functions have the structure
\begin{equation}
	\begin{split}
		\ytableausetup{centertableaux, smalltableaux} \quad 
		\tau_{\bm{u}}^{(a)}(\bm{t})\propto
		&\ c^{(a)}_{(1,0)} \ydiagram[*(lightgray)]{1} +  c^{(a)}_{(3,0)}  \ydiagram[*(lightgray)]{3}+  c^{(a)}_{(5,0)}  \ydiagram[*(lightgray)]{5} + \cdots \\[2ex]
		+&\ c^{(a)}_{(3,2)} \ydiagram[*(lightgray)]{2,3} + c^{(a)}_{(5,2)} \ydiagram[*(lightgray)]{2,5} + c^{(a)}_{(7,0)} \ydiagram[*(lightgray)]{2,7} +  \cdots
	\end{split}	
\end{equation}
Observe that this case \(M=2\) is defined Young diagrams with an odd number of boxes.

Assuming, once again, a maximal number \(\lambda_1^{max}\in 2\mathbb{N}+1\) of boxes in the first row, it is easy to see that there are \((\lambda^{max}_1 +1)/2\) diagrams in the first row. Furthermore, since the maximal number of boxes in the second row is \(\lambda_2^{max}=\lambda_1^{max}-1\), there are \((\lambda_2^{max}+2)/2=(\lambda_1^{max}+1)/2\) diagrams for fixed \(\lambda_1^{max}\). Repeating this idea to all cases \(1\leq \lambda_1<\lambda_1^{max}\), it is easy to conclude that the total number of diagrams for fixed \(\lambda_1^{max}\) is
\be 
\label{eq:ndiagrams}
\begin{split}
	\# \textrm{Diagrams} & = \sum^{\lambda_1^{max}}_{\substack{k=1\\ k\in 2\mathbb{N}+1}}\frac{k+1}{2} = \sum^{(\lambda_1^{max}+1)/2}_{j=1} j  \\
	& = \frac{1}{8}(\lambda^{max}_1+1)(\lambda^{max}_1+3) 
\end{split}\; ,
\ee
where we have written the sum as an arithmetic progression. 

Observe that we can organize the \(M=2\) case in the matrix structure
\be 
\begin{pmatrix}
	s_{(1,0)} & s_{(3,0)} & s_{(5,0)} & s_{(7,0)} & \cdots \\
	s_{(3,2)} & s_{(5,2)} & s_{(7,2)} & s_{(9,2)} & \cdots \\
	s_{(5,4)} & s_{(7,4)} & s_{(9,4)} & s_{(11,4)} & \cdots \\
	s_{(7,6)} & s_{(9,6)} & s_{(11,6)} & s_{(13,6)} & \cdots\\
	\vdots & & \cdots & & \ddots
\end{pmatrix}
\ee
where \(s_{(\lambda_1,\lambda_2)}\) denotes the Schur polynomial for the partition \(\bm{\lambda}=(\lambda_1, \lambda_2)\). Therefore, expression~(\ref{eq:ndiagrams}) gives the number of elements in and above the \(\tfrac{(\lambda_1^{max}+1)}{2}\textsuperscript{th}\) anti-diagonal, see Fig.~\ref{fig:m2case}.
\begin{figure}[htb!]
	\centering
	\begin{subfigure}[b]{0.2\textwidth}
		\includegraphics[width=\textwidth,angle=180]{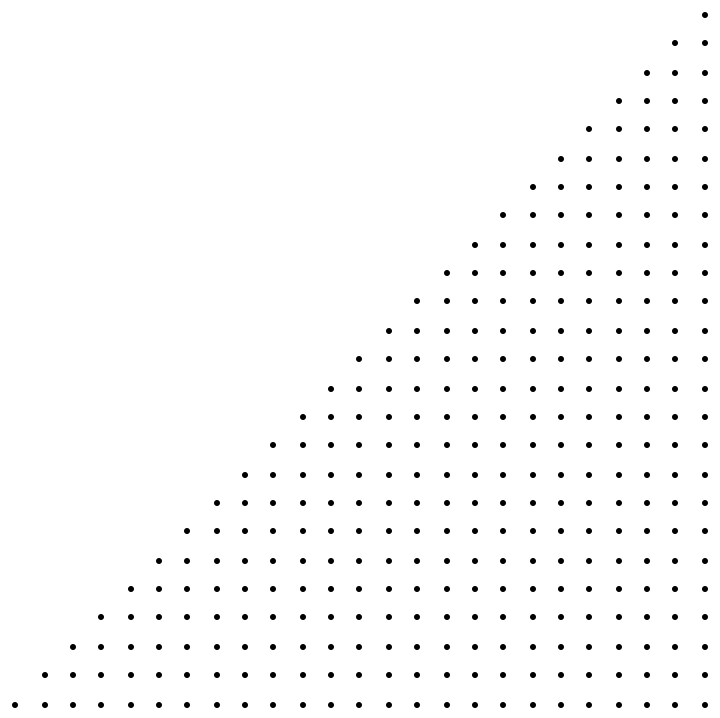}
		\caption{\(\lambda_1^{max}=49\).}
	\end{subfigure}	\hspace{2truecm}
	\begin{subfigure}[b]{0.2\textwidth}
		\includegraphics[width=\textwidth,angle=180]{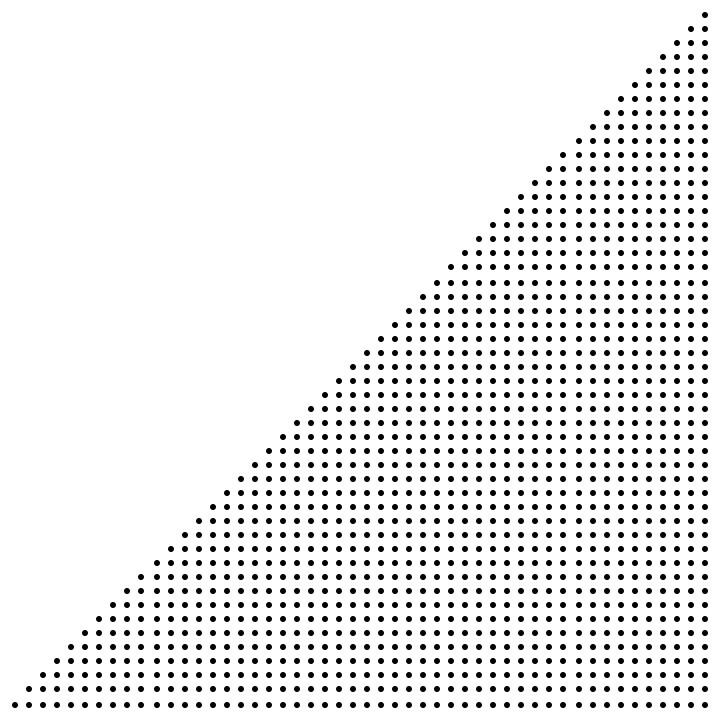}
		\caption{\(\lambda_1^{max}=99\).}
	\end{subfigure}	
	\caption{Array structure for the spin chain in the \(M=2\) magnon sector.}
	\label{fig:m2case}
\end{figure}
\medskip 

\paragraph{\(\bullet\ \bm{M=3}\):} We have the diagrams
\be 
\bm{\lambda} = (\lambda_1, \lambda_2, \lambda_3)\qquad \text{with} \quad \lambda_1>\lambda_2>\lambda_3\; , \quad \lambda_1 \in 2\mathbb{N} \; , \quad \lambda_2 \in 2\mathbb{N}+1\; , \quad \lambda_3 \in 2\mathbb{N}\; .
\ee
The minimal diagram in this configuration is \(\lambda^{min}=(2,1,0)\). The tau functions have the form
\begin{equation}
	\begin{split}
		\ytableausetup{centertableaux, smalltableaux} 
		\tau_{\bm{u}}^{(a)}(\bm{t})\propto
		&\ c^{(a)}_{(2,1,0)} \ydiagram[*(lightgray)]{1,2} + c^{(a)}_{(4,1,0)} \ydiagram[*(lightgray)]{1,4} + c^{(a)}_{(6,1,0)} \ydiagram[*(lightgray)]{1,6} + \cdots \\
		& \ c^{(a)}_{(4,3,0)} \ydiagram[*(lightgray)]{3,4} + c^{(a)}_{(6,1,0)} \ydiagram[*(lightgray)]{3,6} + c^{(a)}_{(8,3,0)} \ydiagram[*(lightgray)]{3,8} + \cdots \\
		& \ c^{(a)}_{(4,3,2)} \ydiagram[*(lightgray)]{2,3,4}+ c^{(a)}_{(6,3,2)} \ydiagram[*(lightgray)]{2,3,6} +  c^{(a)}_{(8,3,2)} \ydiagram[*(lightgray)]{2,3,8} + \cdots 
	\end{split}
\end{equation}

Using our condition of a maximal \(\lambda_1^{max}\in 2\mathbb{N}\) once again, it is easy to see that there are \(\lambda_1^{max}/2\) diagrams of the form \((\lambda_1,1,0)\), with \(2\leq \lambda_1\leq \lambda_1^{max}\). Using the same argument, the maximal number of boxes in the second row is \(\lambda_2^{max}=\lambda_1^{max}-1\), and it implies the existence of \((\lambda_2^{max}+1)/2=\lambda_1^{max}/2\) diagrams of the form \((\lambda_1^{max}, \lambda_2,0)\) for fixed \(\lambda_1^{max}\) and \(\lambda_2\geq 1\). Finally, if we fix the length of the second row \(\lambda_2^{max}\), there are 
\((\lambda_3^{max}+2)/2=(\lambda_2^{max}+1)/2\) diagrams of the form \((\lambda_1, \lambda_2^{max} ,\lambda_3)\) for \(0\leq \lambda_3\leq \lambda_2-1\). For a maximal \(\lambda_1^{max}\), the total number of Young diagrams is
\be 
\begin{split}
	\# \textrm{Diagrams} & = \sum^{\lambda_1^{max}}_{\substack{k=2\\ k\in 2\mathbb{N}}} \sum^{k-1}_{j=1} \frac{j+1}{2} = \sum^{\lambda_1^{max}}_{\substack{k=2\\ k\in 2\mathbb{N}}} \sum^{k/2}_{j=1} j  \\
	& = \frac{1}{8} \sum^{\lambda_1^{max}}_{\substack{k=2\\ k\in 2\mathbb{N}}} k(k+2) = \frac{1}{48}\lambda_1^{max}((\lambda_1^{max})^2 + 6\lambda_1^{max} +8)
\end{split}\; .
\ee 
We also have an array structure
\be 
\begin{pmatrix}
	s_{(2,1,\lambda_3)} & s_{(2,1,\lambda_3)} & s_{(4,1,\lambda_3)} & s_{(6,1,\lambda_3)} & \cdots \\
	s_{(4,3,\lambda_3)} & s_{(6,3,\lambda_3)} & s_{(8,3,\lambda_3)} & s_{(10,3,\lambda_3)} & \cdots \\
	s_{(6,5,\lambda_3)} & s_{(8,5,\lambda_3)} & s_{(10,5,\lambda_3)} & s_{(12,5,\lambda_3)} & \cdots \\
	s_{(8,7,\lambda_3)} & s_{(10,7,\lambda_3)} & s_{(12,7,\lambda_3)} & s_{(14,7,\lambda_3)} & \cdots\\
	\vdots & & \cdots & & \ddots
\end{pmatrix}\qquad\text{for}\ \  \lambda_3=0, 2, 4, \dots 
\ee
Finally, we can also represent this structure diagrammatically as Fig.~\ref{fig:m3case}.
\begin{figure}[htb!]
	\centering
	\begin{subfigure}[b]{0.2\textwidth}
		\includegraphics[width=\textwidth]{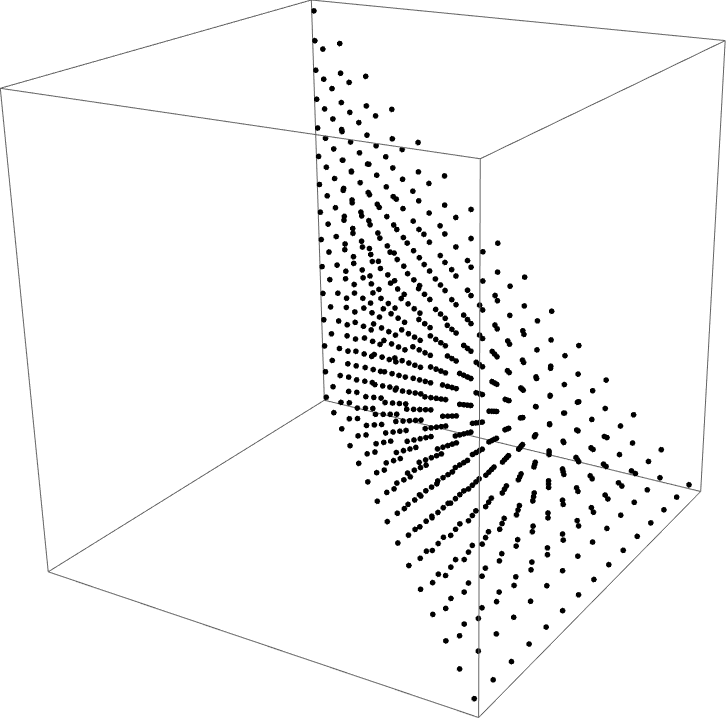}
		\caption{\(\lambda_1^{max}=30\).}
	\end{subfigure}	\hspace{2truecm}
	\begin{subfigure}[b]{0.2\textwidth}
		\includegraphics[width=\textwidth]{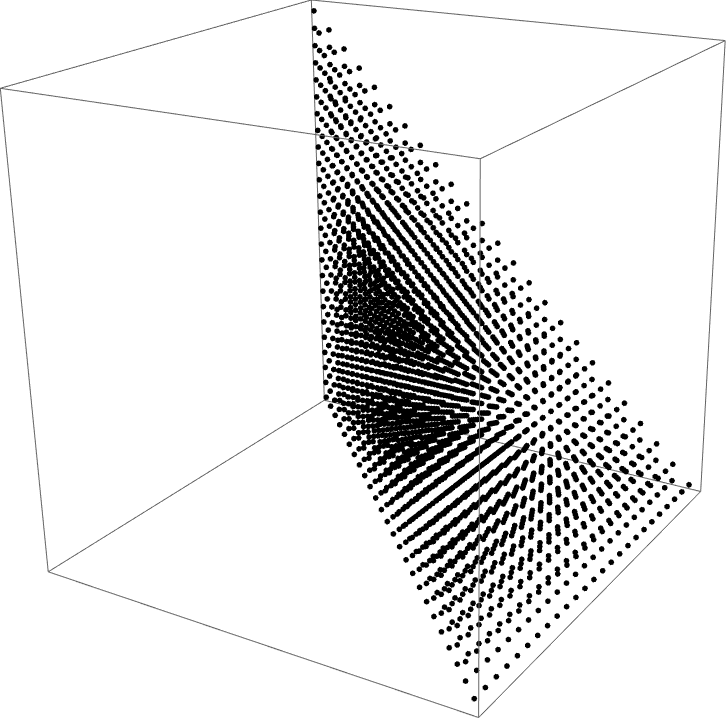}
		\caption{\(\lambda_1^{max}=50\).}
	\end{subfigure}	
	\caption{Array structure for the spin chain in the \(M=3\) magnon sector.}
	\label{fig:m3case}
\end{figure}
\medskip 

\paragraph{\(\bullet\ \bm{M}\):} For general \(M\), the minimal Young diagram is given by
\be 
\bm{\lambda}^{min} = (M-1, M-2, M-3, \dots, 2, 1, 0)\; .
\ee
Using the arguments above, we can delineate the argument for the number of boxes for configurations limited by \(\lambda_1^{max}\) in the \(M\)-magnon sector as
\be 
\# \textrm{Diagrams} = \sum_{k_1=M-1}^{\lambda_1^{max}}  \sum_{k_2=M-2}^{k_1} \cdots \sum_{k_{M-1}=1}^{k_{M-2}} \frac{k_{M-1}+1}{2}\; ,
\ee
where we have use initial condition \(k_0=1\).

\bibliographystyle{utphys}
\bibliography{library.bib}

\providecommand{\href}[2]{#2}\begingroup\raggedright\begin{thebibliography}{10}

\bibitem{Wu:1975mw}
T.~T. Wu, B.~M. McCoy, C.~A. Tracy, and E.~Barouch, ``{Spin spin correlation
  functions for the two-dimensional Ising model: Exact theory in the scaling
  region},'' \href{http://dx.doi.org/10.1103/PhysRevB.13.316}{{\em Phys. Rev.
  B} {\bf 13} (1976)  316--374}.

\bibitem{Its:1992bj}
A.~R. Its, A.~G. Izergin, V.~E. Korepin, and N.~A. Slavnov, ``{Temperature
  correlations of quantum spins},''
  \href{http://dx.doi.org/10.1103/PhysRevLett.70.1704}{{\em Phys. Rev. Lett.}
  {\bf 70} (1993)  1704--1708}, \href{http://arxiv.org/abs/hep-th/9212135}{{\tt
  arXiv:hep-th/9212135}}. [Erratum: Phys.Rev.Lett. 70, 2357 (1993)].

\bibitem{Korepin:1993kvr}
V.~E. Korepin, N.~M. Bogoliubov, and A.~G. Izergin,
  \href{http://dx.doi.org/10.1017/CBO9780511628832}{{\em {Quantum Inverse
  Scattering Method and Correlation Functions}}}.
\newblock Cambridge Monographs on Mathematical Physics. Cambridge University
  Press, Cambridge, 1993.

\bibitem{Korepin2000}
V.~Korepin and P.~Zinn-Justin, ``Thermodynamic limit of the six-vertex model
  with domain wall boundary conditions,''
  \href{http://dx.doi.org/10.1088/0305-4470/33/40/304}{{\em Journal of Physics
  A: Mathematical and General} {\bf 33} (2000) no.~40, 7053–7066}.
  \url{http://dx.doi.org/10.1088/0305-4470/33/40/304}.

\bibitem{Alexandrov:2011aa}
A.~Alexandrov, V.~Kazakov, S.~Leurent, Z.~Tsuboi, and A.~Zabrodin, ``{Classical
  tau-function for quantum spin chains},''
  \href{http://dx.doi.org/10.1007/JHEP09(2013)064}{{\em JHEP} {\bf 09} (2013)
  064}, \href{http://arxiv.org/abs/1112.3310}{{\tt arXiv:1112.3310 [math-ph]}}.

\bibitem{Foda2009a}
O.~Foda, M.~Wheeler, and M.~Zuparic, ``Domain wall partition functions and
  kp,'' \href{http://dx.doi.org/10.1088/1742-5468/2009/03/p03017}{{\em Journal
  of Statistical Mechanics: Theory and Experiment} {\bf 2009} (2009) no.~03,
  P03017}. \url{http://dx.doi.org/10.1088/1742-5468/2009/03/P03017}.

\bibitem{Foda2009}
O.~Foda, M.~Wheeler, and M.~Zuparic, ``Xxz scalar products and kp,''
  \href{http://dx.doi.org/10.1016/j.nuclphysb.2009.04.019}{{\em Nuclear Physics
  B} {\bf 820} (2009) no.~3, 649–663}.
  \url{http://dx.doi.org/10.1016/j.nuclphysb.2009.04.019}.

\bibitem{Slavnov1989}
N.~A. Slavnov, ``{Calculation of scalar products of wave functions and form
  factors in the framework of the alcebraic Bethe ansatz},''
  \href{http://dx.doi.org/10.1007/bf01016531}{{\em Theoretical and Mathematical
  Physics} {\bf 79} (1989) no.~2, 502--508}.

\bibitem{Miwa2000}
T.~Miwa, M.~Jinbo, M.~Jimbo, and E.~Date, {\em Solitons: Differential
  Equations, Symmetries and Infinite Dimensional Algebras}.
\newblock Cambridge Tracts in Mathematics. Cambridge University Press, 2000.
\newblock \url{https://books.google.ch/books?id=kQDw1ZcqLjUC}.

\bibitem{Nepomechie:2016ejv}
R.~I. Nepomechie and R.~A. Pimenta, ``{Universal Bethe ansatz solution for the
  Temperley\textendash{}Lieb spin chain},''
  \href{http://dx.doi.org/10.1016/j.nuclphysb.2016.04.045}{{\em Nucl. Phys. B}
  {\bf 910} (2016)  910--928}, \href{http://arxiv.org/abs/1601.04378}{{\tt
  arXiv:1601.04378 [math-ph]}}.

\bibitem{Nepomechie:2016ruv}
R.~I. Nepomechie and R.~A. Pimenta, ``{Algebraic Bethe ansatz for the
  Temperley\textendash{}Lieb spin-1 chain},''
  \href{http://dx.doi.org/10.1016/j.nuclphysb.2016.04.044}{{\em Nucl. Phys. B}
  {\bf 910} (2016)  885--909}, \href{http://arxiv.org/abs/1601.04328}{{\tt
  arXiv:1601.04328 [math-ph]}}.

\bibitem{Temperley1971}
H.~N.~V. Temperley, E.~H. Lieb, and S.~F. Edwards, ``Relations between the
  `percolation' and `colouring' problem and other graph-theoretical problems
  associated with regular planar lattices: some exact results for the
  `percolation' problem,'' \href{http://dx.doi.org/10.1098/rspa.1971.0067}{{\em
  Proceedings of the Royal Society of London. A. Mathematical and Physical
  Sciences} {\bf 322} (1971) no.~1549, 251--280}.

\bibitem{Abramsky2009}
S.~Abramsky, ``Temperley-lieb algebra: From knot theory to logic and
  computation via quantum mechanics,''
  \href{http://arxiv.org/abs/0910.2737}{{\tt arXiv:0910.2737 [quant-ph]}}.

\bibitem{Doikou:2009xq}
A.~Doikou, S.~Evangelisti, G.~Feverati, and N.~Karaiskos, ``{Introduction to
  Quantum Integrability},''
  \href{http://dx.doi.org/10.1142/S0217751X10049803}{{\em Int. J. Mod. Phys. A}
  {\bf 25} (2010)  3307--3351}, \href{http://arxiv.org/abs/0912.3350}{{\tt
  arXiv:0912.3350 [math-ph]}}.

\bibitem{Pasquier:1989kd}
V.~Pasquier and H.~Saleur, ``{Common Structures Between Finite Systems and
  Conformal Field Theories Through Quantum Groups},''
  \href{http://dx.doi.org/10.1016/0550-3213(90)90122-T}{{\em Nucl. Phys. B}
  {\bf 330} (1990)  523--556}.

\bibitem{Kitanine:2007bi}
N.~Kitanine, K.~K. Kozlowski, J.~M. Maillet, G.~Niccoli, N.~A. Slavnov, and
  V.~Terras, ``{Correlation functions of the open XXZ chain I},''
  \href{http://dx.doi.org/10.1088/1742-5468/2007/10/P10009}{{\em J. Stat.
  Mech.} {\bf 0710} (2007)  P10009}, \href{http://arxiv.org/abs/0707.1995}{{\tt
  arXiv:0707.1995 [hep-th]}}.

\bibitem{Wang:2002paz}
Y.-S. Wang, ``{The scalar products and the norm of Bethe eigenstates for the
  boundary XXX Heisenberg spin-1/2 finite chain},''
  \href{http://dx.doi.org/10.1016/S0550-3213(01)00610-1}{{\em Nucl. Phys. B}
  {\bf 622} (2002)  633--649}.

\bibitem{Kharchev:1991cy}
S.~Kharchev, A.~Marshakov, A.~Mironov, A.~Morozov, and A.~Zabrodin, ``{Towards
  unified theory of 2-d gravity},''
  \href{http://dx.doi.org/10.1016/0550-3213(92)90521-C}{{\em Nucl. Phys. B}
  {\bf 380} (1992)  181--240}, \href{http://arxiv.org/abs/hep-th/9201013}{{\tt
  arXiv:hep-th/9201013}}.

\bibitem{Kharchev:1991cu}
S.~Kharchev, A.~Marshakov, A.~Mironov, A.~Morozov, and A.~Zabrodin,
  ``{Unification of all string models with C \ensuremath{<} 1},''
  \href{http://dx.doi.org/10.1016/0370-2693(92)91595-Z}{{\em Phys. Lett. B}
  {\bf 275} (1992)  311--314}, \href{http://arxiv.org/abs/hep-th/9111037}{{\tt
  arXiv:hep-th/9111037}}.

\bibitem{Morozov:1994hh}
A.~Morozov, ``{Integrability and matrix models},''
  \href{http://dx.doi.org/10.1070/PU1994v037n01ABEH000001}{{\em Phys. Usp.}
  {\bf 37} (1994)  1--55}, \href{http://arxiv.org/abs/hep-th/9303139}{{\tt
  arXiv:hep-th/9303139}}.

\bibitem{Marshakov:2006rh}
A.~Marshakov, ``{Matrix models, complex geometry and integrable systems. I.},''
  \href{http://dx.doi.org/10.1007/s11232-006-0065-x}{{\em Theor. Math. Phys.}
  {\bf 147} (2006)  583--636}, \href{http://arxiv.org/abs/hep-th/0601212}{{\tt
  arXiv:hep-th/0601212}}.

\bibitem{Zabrodin2018}
A.~Zabrodin, ``Lectures on nonlinear integrable equations and their
  solutions,'' \href{http://arxiv.org/abs/1812.11830}{{\tt arXiv:1812.11830
  [math-ph]}}.

\bibitem{Zinn-Justin:2002fzw}
P.~Zinn-Justin, ``{HCIZ integral and 2-D Toda lattice hierarchy},''
  \href{http://dx.doi.org/10.1016/S0550-3213(02)00374-7}{{\em Nucl. Phys. B}
  {\bf 634} (2002)  417--432}, \href{http://arxiv.org/abs/math-ph/0202045}{{\tt
  arXiv:math-ph/0202045}}.

\bibitem{Zinn-Justin:2002rai}
P.~Zinn-Justin and J.~B. Zuber, ``{On some integrals over the U(N) unitary
  group and their large N limit},''
  \href{http://dx.doi.org/10.1088/0305-4470/36/12/318}{{\em J. Phys. A} {\bf
  36} (2003)  3173--3194}, \href{http://arxiv.org/abs/math-ph/0209019}{{\tt
  arXiv:math-ph/0209019}}.

\bibitem{Zinnjustin2009}
P.~Zinn-Justin, ``Six-vertex, loop and tiling models: Integrability and
  combinatorics,'' \href{http://arxiv.org/abs/0901.0665}{{\tt arXiv:0901.0665
  [math-ph]}}.

\bibitem{Alexandrov:2012tr}
A.~Alexandrov and A.~Zabrodin, ``{Free fermions and tau-functions},''
  \href{http://dx.doi.org/10.1016/j.geomphys.2013.01.007}{{\em J. Geom. Phys.}
  {\bf 67} (2013)  37--80}, \href{http://arxiv.org/abs/1212.6049}{{\tt
  arXiv:1212.6049 [math-ph]}}.

\bibitem{Orlov:2002jx}
A.~Y. Orlov, ``{Tau functions and matrix integrals},''
  \href{http://arxiv.org/abs/math-ph/0210012}{{\tt arXiv:math-ph/0210012}}.

\bibitem{Takasaki:2010qm}
K.~Takasaki, ``{KP and Toda tau functions in Bethe ansatz},''
  \href{http://arxiv.org/abs/1003.3071}{{\tt arXiv:1003.3071 [math-ph]}}.

\bibitem{Macdonald1998}
I.~Macdonald, {\em Symmetric Functions and Hall Polynomials}.
\newblock Oxford classic texts in the physical sciences. Clarendon Press, 1998.
\newblock \url{https://books.google.ch/books?id=srv90XiUbZoC}.

\bibitem{Kitanine1999}
N.~Kitanine, J.~Maillet, and V.~Terras, ``Form factors of the xxz heisenberg
  finite chain,'' \href{http://dx.doi.org/10.1016/s0550-3213(99)00295-3}{{\em
  Nuclear Physics B} {\bf 554} (1999) no.~3, 647–678}.
  \url{http://dx.doi.org/10.1016/S0550-3213(99)00295-3}.

\end{thebibliography}\endgroup

\end{document}